\documentclass{article}
\usepackage{amsmath,amsthm,amsfonts}
\usepackage{graphicx}      
\usepackage{url}
\usepackage{xcolor}

\newtheorem{theorem}{Theorem}
\newtheorem{prop}{Proposition}
\newtheorem{lemma}{Lemma}
 \theoremstyle{definition}

\newcommand\RE{\mathbb{R}}

\newcommand\rd{\mathrm{d}}

\DeclareMathOperator{\Co}{Co}

\begin{document}

\title{Observer-based predictor \\ for a SIR model with delays}
\author{Sabine Mondi\'{e} and Fernando Casta\~nos}

\maketitle

\abstract{
We propose an observer for a SIR epidemic model. The observer
is then uplifted into a predictor to compensate for time delays in
the input and the output. Tuning criteria are given for tuning gains
of the predictor, while the estimation-error stability is ensured using
Lyapunov-Krasovskii functionals. The predictor's performance is
evaluated in combination with a time-optimal control. We show that
the predictor nearly recovers the performance level of the delay-free system.
}

\section{Introduction}

Epidemics are a good example of how reality challenges researchers,
offering the opportunity to show the strength of existing techniques
and develop new ones in fields as varied as medicine, biology, computational sciences, 
and mathematical system theory.

Epidemiological models have been primarily used for prediction purposes,
while mitigation policies are usually decided based on exhaustive simulations. 
From the perspective of control theory, an epidemic is viewed as a dynamical
system with controlled variables. Its model is an instrument for designing a control
action that will achieve the desired outcome. Depending on the context, different
assumptions on the model and other control objectives can be formulated. Most works
focus on vaccination or treatment policies, with the goals expressed in an optimal-control
framework~\cite{sharomi2017}. In the context of the current covid 
pandemic, with vaccines unavailable at present, attention is shifting towards intervention
policies based on social distancing measures. Considering the dramatic effect that extended
lockdowns have on people and countries economies, a minimum-time control using social
distancing measures and considering hospital capacity restrictions was recently presented
for the SIR model~\cite{angulo2020}.

The SIR model, which we also consider in this paper, is arguably the simplest epidemiological
model. However, it already exhibits many of the nonlinear characteristics that are present in
more elaborate models. We make the model more realistic by adding features such as inaccurate
and partial state measurements, and input and measurement delays. In recent months, delays in
measurements and policy implementation have proved to be critical in the success or failure of
government strategies. The former correspond to the time taken for the tests to be carried out,
processed, verified, and made available in centralized databases. The latter correspond to the
time it takes the population to adopt restrictions such as quarantine, social distancing habits,
and mask use.
 
Due to the prevalence of delays in the feedback loops of control systems and the associated
detrimental effects on performance, input and output delays have received sustained
attention in the past decades. Some early strategies for compensating the delays are the Smith
controller~\cite{smith1959controller}, the transformation-based reduction 
approach~\cite{artstein1982linear}, and the time-domain predictor-based designs~\cite{manitius1979finite}. 
Based on present state information, predictors based on Cauchy's formula 
provide the state ahead of time~\cite{Bellman1963}. They were formally shown to ensure closed-loop 
stability~\cite{krstic2008boundary}, but their practical implementation reveals instabilities
due to neutral phenomena related to the integrals' discretization in Cauchy's formula.
These issues inspired new proposals such as filtered 
predictors~\cite{mondie2003finite,kharitonov2015predictor}, and truncated
predictors~\cite{zhou2012truncated}. If the present state is not entirely measurable,
it can be replaced by an estimation, provided that the system is observable~\cite{KARAFYLLIS20133623}.

A recent approach to the compensation of delays consists of modifying 
an observer to predict future states. It was first introduced
for the case of full-state information~\cite{Najafi2013} and later extended to the
partial information scenario~\cite{Lechappe2018}. This approach, called observer-predictor, 
suffers from some drawbacks: It loses the exact nature of
predictions obtained with Cauchy's formula and requires the inclusion
of extra sub-predictors designed via LMI techniques. However, it has
significant advantages: The observer has the same structure as the
system (modulo an output injection term), thus avoiding integrals in the 
prediction formulae. Also, it is readily applicable to the case of partial
state information in observable systems, especially when observers are readily
available. Systems with state delays~\cite{ZHOU2017368} or nonlinear 
systems~\cite{Velasco} can be modified easily to successfully compensate for input or output delays.

To tackle the complexity due to partial state availability, delay, and measurement errors,
we resort to a wide array of tools available to specialists in the field of control of 
dynamical systems. For the control, we use a recent optimal law~\cite{angulo2020}.
The objective is not to steer the epidemics towards a desired equilibrium. Instead,
the aim is to track an optimal trajectory. As a result, the dynamics for the estimation
error are time-varying and time-delayed. The stability of such dynamics is addressed from 
both the perspective of classical frequency-domain quasipolynomial
analysis~\cite{Neimark1949,MichielsNiculescu2007}, and from the perspective of time-domain
Lyapunov-Krasowskii analysis~\cite{gu,fridmanE}. In particular, the system's
time-varying nature is taken into account by embedding the system into a model with
polytopic uncertainty~\cite{boyd,he2004}.

In Section~\ref{sec:Problem statement}, we introduce the SIR model and discuss the 
issues we want to overcome. An \emph{ad hoc} change of variable allows designing an 
observer addressing incomplete state information for the delay-free system. In 
Section~\ref{sec:Observer-based},
this observer is developed into an observer-based predictor for
the system with input and output delays. The next two sections are devoted to
tuning the observer: A simple criterion to tune the observer gains is
given in Section~\ref{sec:tune}, and conditions for the stability of the prediction error
dynamics are given in Section~\ref{sec:Stab_linear}. The impact of measurement
errors is discussed in Section~\ref{sec:noise}. We conclude with some remarks. 
We show the validity of our approach by discussing a SIR case study along with the paper, 
which is of interest in its own right.

Allow us to recall some standard notation used in the literature of 
time-delay systems.

\paragraph{Notation.}
$PC([-\eta,0],\RE^{n})$ is the set of piece-wise continuous functions defined
on the interval $[-\eta,0]$. Consider a time-delay differential equation 
\begin{equation}  \label{eq:general}
 \dot{\varepsilon}(\tau) = f(\tau,\varepsilon(\tau),\varepsilon(\tau-\eta(\tau))) \;.
\end{equation}
The time-varying delay, $\eta(\tau)$, is bounded by $0 < \eta(\tau) \le \bar{\eta}$. Given an initial function
$\varphi \in PC([-\eta(0),0],\RE^{n})$ the solution is denoted by $\varepsilon(\tau,\varphi)$. 
The restriction of $\varepsilon(\tau,\varphi)$ on the interval $[\tau-\eta(\tau),\tau]$ is denoted by
\begin{equation*}
 \varepsilon_{\tau}(\varphi) : \theta \mapsto \varepsilon(\tau+\theta,\varphi) \;, 
  \quad \theta \in [-\eta(\tau),0] \;.
\end{equation*}
We will make use of the \emph{trivial function} $0_{\bar{\eta}} : \theta \mapsto 0$, $\theta \in [-\bar{\eta},0]$.
We use the Euclidean norm $\Vert \cdot \Vert$ for vectors and the
corresponding induced norm for matrices. For $\varphi \in PC([-\eta,0],\RE^{n})$
we use the norm 
\begin{equation*}
 \Vert \varphi \Vert_{\eta} = \underset{\theta \in [-\eta,0]} \sup \Vert
  \varphi (\theta )\Vert \;.
\end{equation*}
The notation $Q>0$ means that the symmetric matrix $Q$ is positive definite.

\section{Problem statement} \label{sec:Problem statement}

We consider a state-space SIR model 
\begin{align*}
 \frac{\rd}{\rd t}S(t) &= -\beta(t-h_1) S(t)I(t) \\
 \frac{\rd}{\rd t}I(t) &= (\beta(t-h_1) S(t)-\gamma) I(t) \\
                  y(t) &= I(t-h_2) \;.
\end{align*}
Here, $S > 0$, $I > 0$ denote the susceptible and the infected, respectively.
The model is normalized, hence $S+I \le 1$. The
transmission rate, $\beta \in [\beta_{\min}, \beta_{\max}]$ with 
$\beta_{\max} > \beta_{\min} > 0$, can be controlled by applying social 
distancing measures, but such measures take effect $h_1$ units of time later.
The only information available at time $t$ is the number of infected people
at time $t-h_2$. The recovery/death rate, $\gamma > 0$, and the time delays, 
$h_1, h_2$, are assumed to be known.

The are of course more sophisticated models. It is possible to include exposed 
individuals (infected but not infectious),  to distinguish between symptomatic
and asymptomatic, dead and recovered, etc. However, for epidemics the parameters
of which have large levels of uncertainty, such as covid-19, a simple model with fewer
parameters is preferable, as long as it is able to reproduce the main features
of the epidemics (hospital saturation, lock-down effects, herd immunity, and so 
forth). A simpler model is also preferable when the objective is to devise decision 
strategies, rather that simulating long-term behavior.
 
Using only the history of $y$, we wish to produce predictions $\hat{S}, \hat{I}$
such that 
\begin{equation*}
 \lim_{t \to \infty} (\hat{S}(t)),\hat{I}(t)) = (S(t+h_1),I(t+h_1)) \;.
\end{equation*}
Our motivation is that, if we have a feedback $\beta^*(S,I)$ that is known
to perform correctly on the system without delays, we can set 
$\beta = \beta^*(\hat{S},\hat{I})$ and expect to recover a similar 
performance\footnote{The full analysis would have to be performed, of course.}.

For concreteness, we consider the optimal-time control strategy described 
by Angulo \textit{et al}~\cite{angulo2020}. For the SIR model, the basic (unmitigated)
reproduction number is computed as $R_0 = \beta_{\max}/\gamma$, while the controlled 
(mitigated) reproduction number is $R_c = \beta_{\min}/\gamma$. Suppose that the 
health system capacity of a given city is limited to $I_{\max}$ infected people. 
The strategy that ensures $I(t) \le I_{\max}$ and achieves herd immunity in a minimal time is
\begin{equation} \label{eq:beta_optimal}
 \beta^\star(S,I) = 
  \begin{cases}
   \beta_{\max} & \text{if $I < \Phi(S)$ or $S \le 1/R_0$} \\
   \beta_{\min} & \text{otherwise}
  \end{cases}
\end{equation}
with
\begin{displaymath}
 \Phi(S) = 
  \begin{cases}
   I_{\max} + \frac{1}{R_c}\ln\left(\frac{S}{S^\star}\right) - (S-S^\star) & \text{if $S^\star \le S \le 1$} \\
   I_{\max}                                                                & \text{if $\frac{1}{R_0} \le S \le S^\star$}
  \end{cases}
\end{displaymath}
and
\begin{displaymath}
 S^\star = \min\left\{\frac{1}{R_c},1\right\} \;.
\end{displaymath}

We consider a recovery rate $\gamma = 1/7$ with time units given in days~\cite{angulo2020}. 
For illustration purposes, we consider the case of Mexico City. The number of ICU
beds is such that $I_{\max} = 12.63 \times 10^{-3}$ (see~\cite{angulo2020}). We take the reproduction
numbers as $R_0 = 1.7$ and $R_c = 1.1$~\footnote{Taken from \url{https://epiforecasts.io/covid/posts/national/mexico/},
June 18, 2020.}. This gives
\begin{displaymath}
 \beta_{\max} = 1.7/7 \quad \text{and} \quad \beta_{\max} = 1.1/7 \;.
\end{displaymath}
 
\begin{figure}
\centering
\includegraphics[width=1.0\columnwidth]{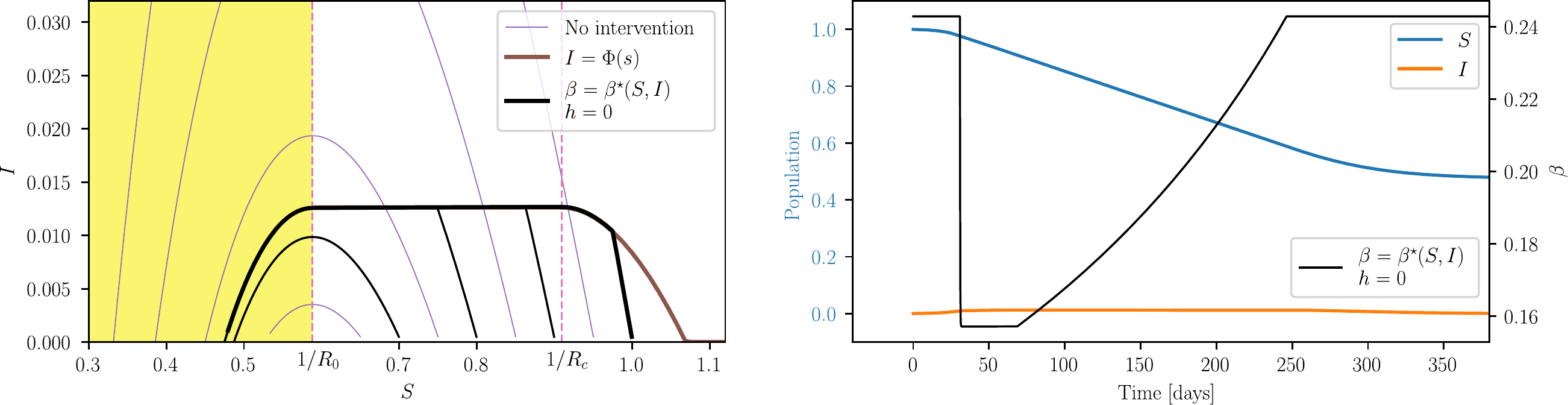}
\caption{Phase-plane of the model under the optimal feedback law~\eqref{eq:beta_optimal} (left). 
 The trajectory with thick line is depicted on the right (blue and orange, scale shown on the
 left axis). The control action is also included (black, scale shown on the right axis).} 
\label{fig:delay_free_full_info}
\end{figure}

The optimal response, achieved with full state-feedback in the absence of delays,
is shown in Fig.~\ref{fig:delay_free_full_info}. The optimal strategy is to 
allow the epidemic to run free until it reaches the sliding curve $I = \Phi(S)$. 
The state is then driven along this curve towards the region of herd immunity
(yellow rectangle) where the intervention finally stops.

\begin{figure}
\centering
\includegraphics[width=1.0\columnwidth]{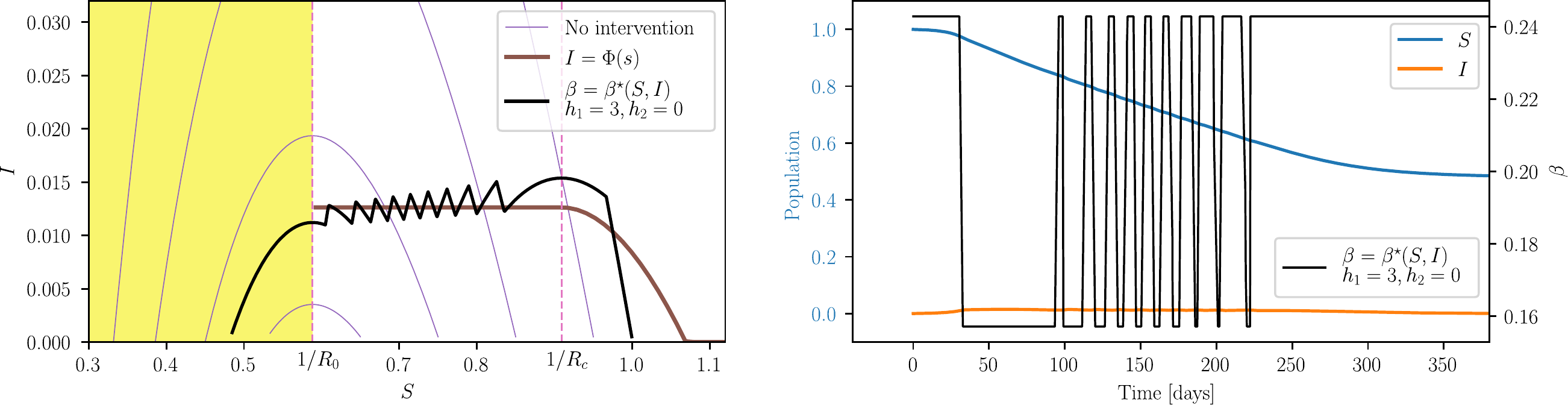}
\caption{Trajectory corresponding to the feedback law~\eqref{eq:beta_optimal} subject
 to a delay of $h_1 = 3$ days that ultimately results in chattering.} 
\label{fig:delayed_input_full_info}
\end{figure}

Suppose now that there is a delay of $h_1 = 3$ days in the control action. 
Fig.~\ref{fig:delayed_input_full_info} confirms the appearance of 
the so-called \emph{chattering} effect, which should not be surprising given the
discontinuous nature of~\eqref{eq:beta_optimal}. We can expect the
performance of the closed-loop system to deteriorate even further when only the 
number of infected people is available for measurement, and 
when such measurements are also subject to important delays. The objective of the 
predictor, developed in the following section, is to mitigate these unfavorable effects.

\section{Observer-based predictor} \label{sec:Observer-based}

To attain our objective we follow the approach in which an observer for the delay-free system 
is constructed in a first step, and then developed into a predictor in a second 
one~\cite{Najafi2013,ZHOU2017368,Villa2017}.

\subsection{Delay-free observer}

We begin by making the temporary assumption $h_1 = h_2=0$. Note that, by
setting $x_1 = \ln(y)$ and $x_2 = S$, we obtain the model 
\begin{align*}
 \frac{\rd}{\rd t}x_1(t) &= \beta(t) x_2(t) - \gamma \\
 \frac{\rd}{\rd t}x_2(t) &= \beta(t) \rho(x(t))
\end{align*}
with $\rho(x) = -x_2 e^{x_1}$. The main advantage over the original 
model is that the first equation is affine in the state. We can then write 
the simple observer 
\begin{align*}
 \frac{\rd}{\rd t}\breve{x}_1(t) &= \beta(t) \left(\breve{x}_2(t) 
  + \alpha_1(x_1(t)-\breve{x}_1(t))\right) - \gamma \\
 \frac{\rd}{\rd t}\breve{x}_2(t) &= \beta(t) \left(\rho(\breve{x}(t)) 
  + \alpha_2(x_1(t)-\breve{x}_1(t)) \right) \;,
\end{align*}
where $\alpha_1, \alpha_2$ will be defined below (Proposition~\ref{prop:delay_free}).
Consider the error $\tilde{x} = x - \breve{x}$. By using the expansion
\begin{displaymath}
 \rho(\breve{x}) = \rho(x) + 
  \begin{pmatrix}
   -x_2e^{x_1} & -e^{x_1} 
  \end{pmatrix}
  \tilde{x}  + \mathcal{O}(\|\tilde{x}\|^2) \;,
\end{displaymath}
we can write the error dynamics as 
\begin{equation} \label{eq:linear_error}
 \frac{\rd}{\rd t}
 \begin{pmatrix}
  \tilde{x}_1(t) \\ 
  \tilde{x}_2(t)
 \end{pmatrix}
  = \beta(t) A(t) 
 \begin{pmatrix}
  \tilde{x}_1(t) \\ 
  \tilde{x}_2(t)
 \end{pmatrix}
 + \mathcal{O}(\|\tilde{x}(t)\|^2)
\end{equation}
with 
\begin{equation*}
 A(t) = 
  \begin{pmatrix}
   -\alpha_1 & 1 \\ 
   -\alpha_2-S(t)I(t) & -I(t)
  \end{pmatrix}
  \;.
\end{equation*}
Note that, since we are linearizing the estimation error around a trajectory 
(rather than an equilibrium), the linearized system is time-varying. Fortunately,
we can ensure its stability with a simple quadratic Lyapunov function.
\begin{prop} \label{prop:delay_free}
 Set
 \begin{equation} \label{eq:alphas}
  \alpha_2 > \frac{\alpha_1^2+1}{4\sqrt{2}\alpha_1-1} \;, \quad \alpha_1 > \frac{1}{4\sqrt{2}} \;.
 \end{equation}
 Then~\eqref{eq:linear_error} is quadratically stable\footnote{Recall that~\eqref{eq:linear_error}
 is quadratically stable if there is a common quadratic Lyapunov function for all
 possible $A(t)$~\cite[Ch. 5]{boyd}.}.
\end{prop}

\begin{proof}
 Consider the candidate Lyapunov function $V(\tilde{x}) = \tilde{x}^\top P \tilde{x}$
 with $P = P^\top > 0$ the solution of the Lyapunov equation
 \begin{displaymath}
  P 
   \begin{pmatrix}
    -\alpha_1 & 1 \\
    -\alpha_2 & 0
   \end{pmatrix}
  +
   \begin{pmatrix}
    -\alpha_1 & -\alpha_2 \\
     1        & 0
   \end{pmatrix}
   P = -
   \begin{pmatrix}
    1 & 0 \\ 0 & 1
   \end{pmatrix} \;,
 \end{displaymath}
 that is,
 \begin{displaymath}
  P = \frac{1}{2\alpha_1 \alpha_2}
  \begin{pmatrix}
   \alpha_2^2+\alpha_2 & -\alpha_1 \alpha_2 \\
    -\alpha_1 \alpha_2 & \alpha_1^2 + \alpha_2 + 1
  \end{pmatrix} \;.
 \end{displaymath}
 The time-derivative of $V$ along the trajectories of~\eqref{eq:linear_error} is 
 \begin{displaymath}
  \dot{V}(\tilde{x}(t)) = -\frac{\beta(t)}{2\alpha_1\alpha_2}\tilde{x}^\top(t) Q(S(t),I(t)) \tilde{x}(t)
   + \mathcal{O}(\|\tilde{x}(t)\|^3)
 \end{displaymath}
 with 
 \begin{displaymath}
  Q(S,I) = 
   \begin{pmatrix}
                         2(1-SI)\alpha_1\alpha_2 & (\alpha_1^2+\alpha_2+1)SI-\alpha_1\alpha_2 I \\
    (\alpha_1^2+\alpha_2+1)SI-\alpha_1\alpha_2 I & 2\left( \alpha_1\alpha_2 + (\alpha_1^2+\alpha_2+1) I \right)
   \end{pmatrix} \;.
 \end{displaymath}
 The restrictions $(S,I) \in [0,1]^2$, $S+I\le1$ imply that $SI \le 1/4$, so the first
 leading principal minor of $Q(S,I)$, $2(1-SI)\alpha_1\alpha_2$, is strictly positive. Regarding 
 the second leading principal minor, we have
 \begin{multline*}
  |Q(S,I)| = \alpha_1^2\alpha_2^2 \left( 4(1-SI)-I^2 \right)
   + \alpha_1\alpha_2(\alpha_1^2+\alpha_2+1)(4-2SI)I \\
    - (\alpha_1^2+\alpha_2+1)^2S^2I^2 \;.
 \end{multline*}
 Using again $SI \le 1/4$ we obtain the bound
 \begin{displaymath}
  |Q(S,I)| \ge 2\alpha_1^2\alpha_2^2 - \frac{1}{16}(\alpha_1^2+\alpha_2+1)^2 \;.
 \end{displaymath}
 Finally, the condition~\eqref{eq:alphas} ensures that $|Q(S,I)| > 0$, so that $V$
 is indeed a Lyapunov function.
\end{proof}

In the original coordinates, the observer takes the form
\begin{equation} \label{eq:observer}
\begin{aligned}
 \frac{\rd}{\rd t}\breve{S}(t) &= -\beta(t)\left( \breve{S}(t)\breve{I}(t)
  -\alpha_{2}\ln\left(\frac{y(t)}{\breve{I}(t)}\right) \right) \\ 
 \frac{\rd}{\rd t}\breve{I}(t) &= \left(\beta(t)\breve{S}(t)-\gamma +
  \beta(t)\alpha_{1}\ln\left(\frac{y(t)}{\breve{I}(t)}\right) \right) \breve{I}(t) \;.
\end{aligned}
\end{equation}

\begin{figure}
\centering
\includegraphics[width=1.0\columnwidth]{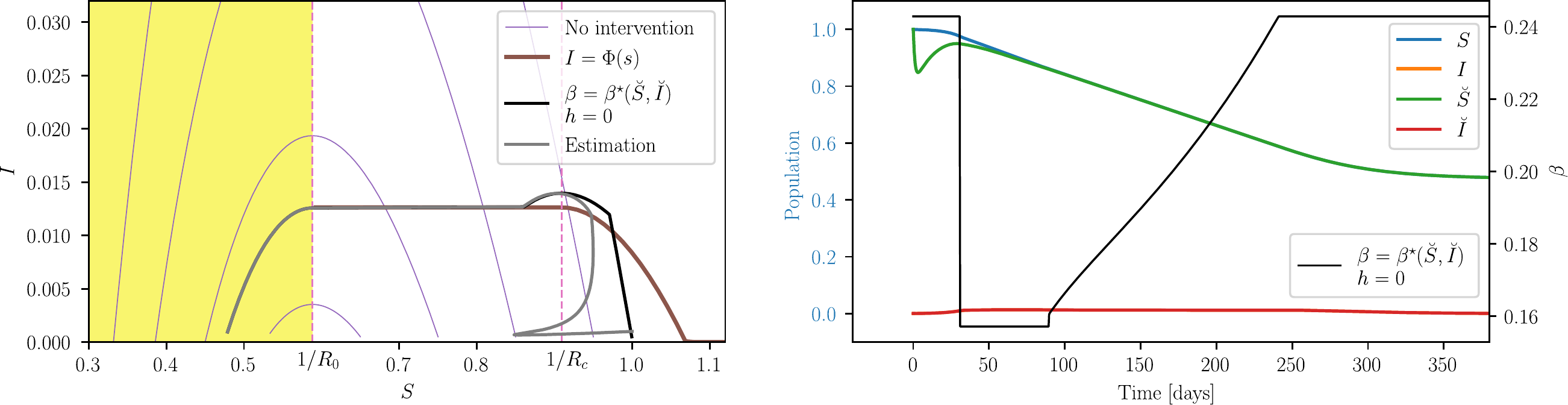}
\caption{True and estimated trajectories  under the feedback law $\beta = \beta^\star(\breve{S},\breve{I})$,
 with no delays.} 
\label{fig:delay_free_partial_info}
\end{figure}

A simulation of the observer's performance is shown in Fig.~\ref{fig:delay_free_partial_info}.
The observer gains, $\alpha^\top = \begin{pmatrix} \alpha_1 & \alpha_2 \end{pmatrix} = 
\begin{pmatrix} 4 & 1 \end{pmatrix}$
were chosen to satisfy~\eqref{eq:alphas}. At first, the incidence of infection exceeds 
$I_{\max}$ by about 11\,\%, but then the epidemics behave as desired.

\begin{figure}
\centering
\includegraphics[width=1.0\columnwidth]{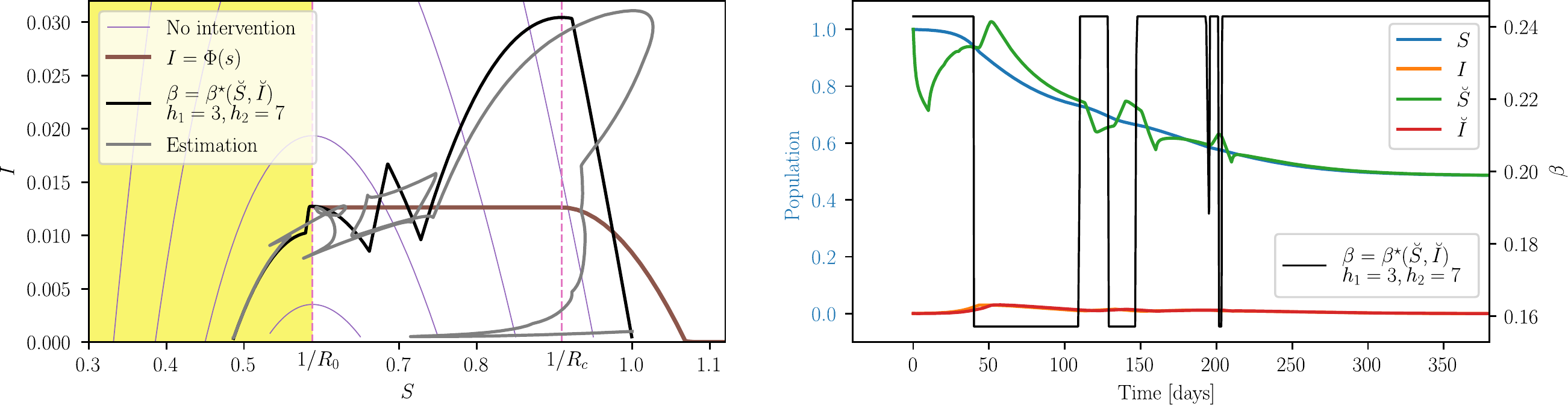}
\caption{True and estimated trajectories under the feedback law $\beta = \beta^\star(\breve{S},\breve{I})$
 with $h_1 = 3$ and $h_2 = 7$ days.} 
\label{fig:delays_observer}
\end{figure}

Consider again the input delay $h_1 = 3$ days, and suppose now there is a measurement delay of
$h_2 = 7$ days~\cite{capistran2020}. The combined effect of both delays is disastrous.
As illustrated in Fig.~\ref{fig:delays_observer}, the hospital capacity is exceeded by 
140\,\%. This motivates the design of the predictor presented in the sequel.

\subsection{Predictor}

Let us rewrite the observer in the original coordinates and remove the 
zero-delays assumption. This gives the predictor
\begin{equation} \label{eq:predictor}
\begin{aligned}
 \frac{\rd}{\rd t}\hat{S}(t) &= -\beta(t)\left( \hat{S}(t)\hat{I}(t)
  -\alpha_{2}\ln\left(\frac{y(t)}{\hat{I}(t-h)}\right) \right) \\ 
 \frac{\rd}{\rd t}\hat{I}(t) &= \left(\beta(t)\hat{S}(t)-\gamma +
  \beta(t)\alpha_{1}\ln\left(\frac{y(t)}{\hat{I}(t-h)}\right) \right) \hat{I}(t)
\end{aligned}
\end{equation}
with $h = h_1+h_2$. Considering that $y(t)=I(t-h_{2})$, the error variables 
\begin{equation} \label{eta_def}
\begin{aligned}
 \tilde{x}_{1}(t) &= \ln \left( \frac{I(t)}{\hat{I}(t-h_{1})}\right)  \\
 \tilde{x}_{2}(t) &= S(t)-\hat{S}(t-h_{1})
\end{aligned}
\end{equation}
evolve according to the dynamics 
\begin{equation} \label{observation error general}
\begin{aligned}
 \frac{\rd}{\rd t}\tilde{x}_{1}(t) &= \beta(t-h_{1})\left(S(t)-\hat{S}(t-h_{1})
  -\alpha_{1}\ln \left( \frac{I(t-h)}{\hat{I}(t-h-h_{1})}\right) \right) \\ 
 \frac{\rd}{\rd t}\tilde{x}_{2}(t) &= \beta(t-h_{1})\Bigg(-\alpha_{2}
  \ln \left( \frac{I(t-h)}{\hat{I}(t-h-h_{1})}\right) + \hat{S}(t-h_{1})\hat{I}(t-h_{1}) - S(t)I(t)\Bigg)
\end{aligned} \;.
\end{equation}
Since $\hat{I}(t-h_{1})=e^{-\tilde{x}_{1}(t)}I(t)$ and $\hat{S}(t-h_1) = S(t)-\tilde{x}_2(t)$, we have 
\begin{equation} \label{prediction error system}
\begin{aligned}
 \frac{\rd}{\rd t}\tilde{x}_{1}(t) &= \beta(t-h_{1})\left(-\alpha_{1}\tilde{x}_{1}(t-h)+\tilde{x}_{2}(t)\right) \\ 
 \frac{\rd}{\rd t}\tilde{x}_{2}(t) &= \beta(t-h_{1})\left(-\alpha_{2}\tilde{x}_{1}(t-h)+\psi (S(t),I(t),\tilde{x}(t))\right)
\end{aligned}
\end{equation}
with
\begin{equation} \label{eq:psi}
 \psi(S,I,\tilde{x}) = ((S-\tilde{x}_{2})e^{-\tilde{x}_{1}} - S)I \;.
\end{equation}
System (\ref{prediction error system}) can be written in the general form
\begin{equation} \label{system}
 \frac{\rd}{\rd t}\tilde{x}(t) = \beta(t-h_{1})\left[ A_0(t)\tilde{x}(t) + A_{1}\tilde{x}(t-h) + G(t,\tilde{x}(t))\tilde{x}(t) \right]
\end{equation}
where the matrices are defined by
\begin{equation} \label{Anominal}
 A_0(t) = 
  \begin{pmatrix}
           0 &  1 \\ 
   -S(t)I(t) & -I(t)
  \end{pmatrix} \;, \quad
 A_1 =  
  \begin{pmatrix}
   -\alpha_{1} & 0 \\ 
   -\alpha_{2} & 0
  \end{pmatrix}
\end{equation}
and
\begin{displaymath}
 G(t,\tilde{x}) = 
  \begin{pmatrix}
                0 & 0 \\ 
   G_{21}(t,\tilde{x}) & G_{22}(t,\tilde{x})
  \end{pmatrix} 
  = \mathcal{O}(\Vert\tilde{x}\Vert)
\end{displaymath}
with 
\begin{equation} \label{g}
\begin{aligned}
 G_{21}(t,\tilde{x}) &= -S(t)I(t)\frac{e^{-\tilde{x}_{1}} + \tilde{x}_{1} - 1}{\tilde{x}_{1}} \\
 G_{22}(t,\tilde{x}) &= -I(t)(e^{-\tilde{x}_{1}}-1)
\end{aligned} \;.
\end{equation}

Since $\beta$ multiplies all the right-hand side of~\eqref{system},
we can do away with it by rescaling time, much in the spirit of perturbation theory~\cite[Ch. 10]{khalil}. 
Define the new time-scale 
\begin{displaymath}
 \tau = g(t) = \int_0^t \beta(s-h_{1})\rd s \;.
\end{displaymath}
Since $\beta$ is strictly positive, $g$ is strictly increasing, invertible and $\tau$ is
indeed a time-scale. Let us define the new state 
$\varepsilon(\tau) = \tilde{x}(g^{-1}(\tau))$ and note that, by the Inverse Function Theorem,
we have 
\begin{equation} \label{eq:der_g}
 \frac{\rd}{\rd \tau} g^{-1}(\tau) = \frac{1}{\frac{\rd}{\rd t}g(t)}\Bigg|_{t=g^{-1}(\tau)} 
                                    = \frac{1}{\beta(t-h_1)}\Bigg|_{t=g^{-1}(\tau)} \;.
\end{equation}
Applying the chain rule and~\eqref{eq:der_g}, we see that the new state evolves
according to
\begin{displaymath}
 \frac{\rd}{\rd \tau} \varepsilon(\tau) = \frac{1}{\beta(t-h_1)}\frac{\rd}{\rd t}\tilde{x}(t)\Big|_{t=g^{-1}(\tau)} \;,
\end{displaymath}
that is,
\begin{equation} \label{eq:system_tau}
 \dot{\varepsilon}(\tau) = B_0(\tau)\varepsilon(\tau) + B_{1}\varepsilon(\tau-\eta(\tau)) 
 + H(\tau,\varepsilon(\tau))\varepsilon(\tau)
\end{equation}
with
\begin{displaymath}
 B_0(\tau) = A_0(g^{-1}(\tau)) \;, \quad B_1 = A_1 \;, \quad H(\tau,\varepsilon) = G(g^{-1}(\tau),\varepsilon)
\end{displaymath}
and
\begin{displaymath}
 \eta(\tau) = \tau - g(g^{-1}(\tau)-h) \;.
\end{displaymath}
The last equation follows from the condition
\begin{displaymath}
 \varepsilon(\tau-\eta(\tau)) = \varepsilon(g(g^{-1}(\tau)-h)) = \tilde{x}(g^{-1}(\tau) - h)) \;.
\end{displaymath}
Observe that $\beta_{\min}h \le \tau - g(t-h) = \int_{t-h}^t \beta(s-h_1)\rd s \le \beta_{\max}h$,
so the time-varying delay is bounded by 
\begin{displaymath}
 \beta_{\min} h\le \eta(\tau) \le \beta_{\max}h = \bar{\eta} \;.
\end{displaymath}

\section{Tuning the predictor gains} \label{sec:tune}

There are two main difficulties in establishing the stability of~\eqref{eq:system_tau}:
The time-varying nature of $B_0$, and the presence of the delayed state. The former
difficulty will be addressed by formulating~\eqref{eq:system_tau} in the framework
of polytopic differential inclusions~\cite[Ch. 5]{boyd}. In order to do so, we will focus 
on the state-space rectangle $[0,1]\times[0,\bar{I}]$, where $I_{\max} \le \bar{I} \le 1$.
When the state is restricted to such rectangle, $B_0$ varies within a fixed polytope of matrices, i.e.,
\begin{equation} \label{eq:polytope}
 B_0(t) \in \Co\{C_1,C_2,C_3\} \;,
\end{equation}
where
\begin{displaymath}
 C_1 = 
  \begin{pmatrix}
   0 & 1 \\
   0 & 0
  \end{pmatrix}
 \;, \quad
 C_2 = 
  \begin{pmatrix}
   0 & 1 \\
   0 & -\bar{I}
  \end{pmatrix} \;, \quad
 C_3
  \begin{pmatrix}
    0 & 1 \\
   -\bar{I} & -\bar{I}  
  \end{pmatrix}
\end{displaymath}
and $\Co$ stands for convex closure, that is,
\begin{displaymath}
 \Co\{C_1,C_2,C_3\} = \left\{ \sum_{i=1}^3 c_i\cdot C_i \mid c_i \ge 0, 
  \, \sum_{i=1}^3 c_i = 1 \right\} \;.
\end{displaymath}
An obvious necessary condition for the stability of the polytopic
model~\eqref{eq:system_tau}-\eqref{eq:polytope} is the stability of its
linearized vertices,
\begin{displaymath}
 \dot{\varepsilon}(\tau) = C_i\varepsilon(\tau) + B_{1}\varepsilon(\tau-\bar{\eta}) \;, \quad i = 1,2,3.
\end{displaymath}
The characteristic equations of the vertices are
\begin{equation} \label{eq:quasipolynomial}
\begin{aligned}
 p_1(s) &= s^2 + \left( s \alpha_1 + \alpha_2\right)e^{-\bar{\eta}s} \\
 p_2(s) &= s^2 + \left( (s + \bar{I}) \alpha_1 + \alpha_2\right)e^{-\bar{\eta}s} + s\bar{I} \\
 p_3(s) &= s^2 + \left( (s + \bar{I}) \alpha_1 + \alpha_2\right)e^{-\bar{\eta}s} + (s+1)\bar{I}
\end{aligned} \;.
\end{equation}

\begin{figure}
\centering
\includegraphics[width=0.4\columnwidth]{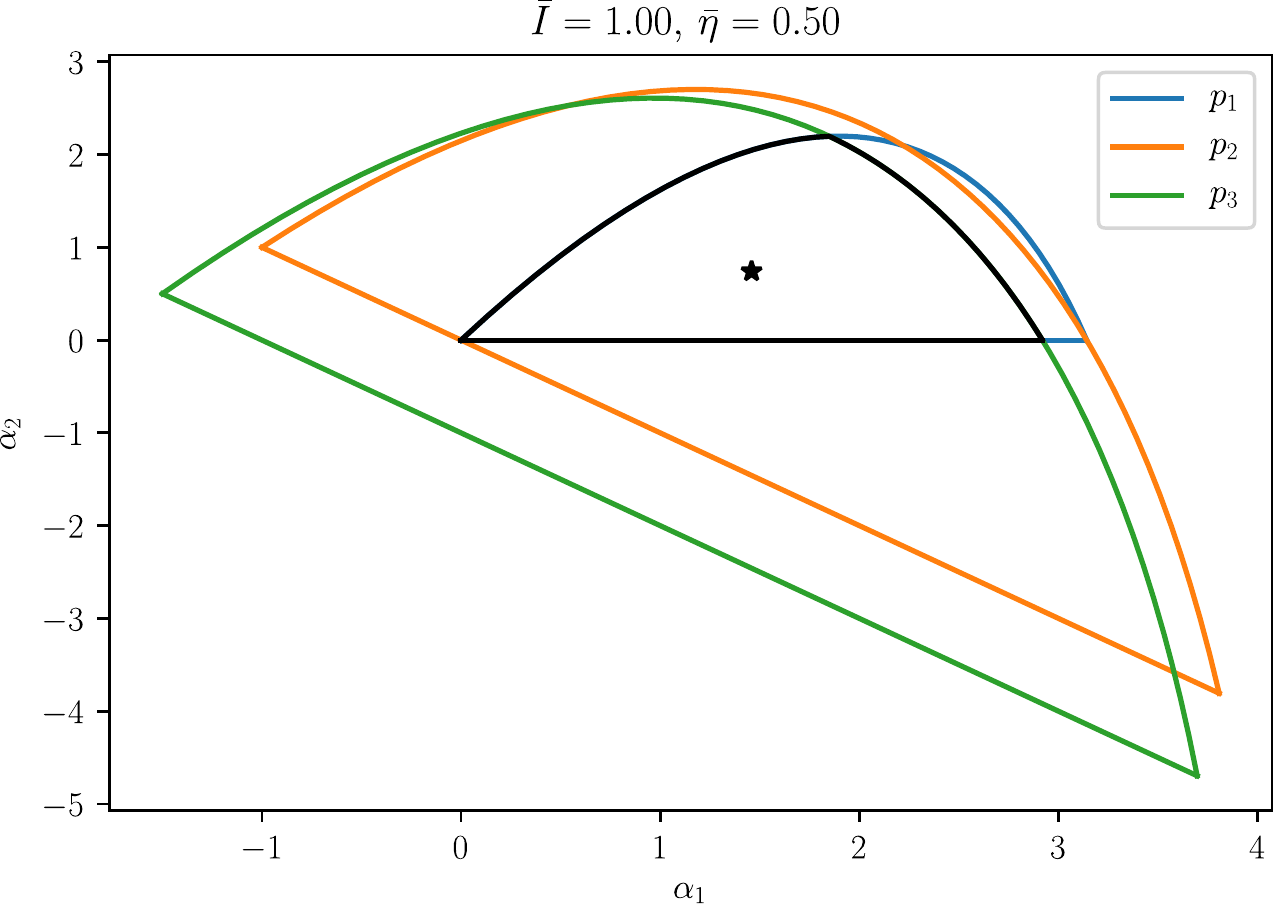}
\caption{Stability/instability boundaries of the quasipolynomials~\eqref{eq:quasipolynomial}
 in the space of parameters $(\alpha_1,\alpha_2)$ for $\bar{I}=1$ and $\bar{\eta}=0.5$. The 
 boundary of the intersection of the three regions is depicted in black.}
\label{fig:vertices_05}
\end{figure}

We will work out the stability/instability boundaries of these quasi\-poly\-no\-mi\-als in
the space of parameters $(\alpha_{1},\alpha_{2})$. According to the
D-partition method~\cite{Neimark1949}, these boundaries correspond to roots
crossing the imaginary axis of the complex plane. When the crossing root is real
($s=0$), the boundaries are
\begin{align*}
 p_1: \quad \alpha_2 &=0 \\
 p_2: \quad \alpha_2 &= -\alpha_1\bar{I} \\
 p_3: \quad \alpha_2 &= -(\alpha_1+1)\bar{I} \;.
\end{align*}
If a crossing root is imaginary $(s=j\omega)$, the boundaries satisfy
\begin{align*}
 p_1&: \quad 
 \begin{pmatrix}
  \omega \sin(\omega \bar{\eta}) & \cos(\omega \bar{\eta}) \\
  \omega \cos(\omega \bar{\eta}) & -\sin(\omega \bar{\eta})
 \end{pmatrix}
 \alpha
  =
 \begin{pmatrix}
  \omega^2 \\ 0
 \end{pmatrix} \\
 p_2&: \quad
 \begin{pmatrix}
  \omega \sin(\omega \bar{\eta}) + \bar{I}\cos(\omega \bar{\eta}) & \cos(\omega \bar{\eta}) \\
  \omega \cos(\omega \bar{\eta}) - \bar{I}\sin(\omega \bar{\eta}) & -\sin(\omega \bar{\eta})
 \end{pmatrix}
 \alpha
  =
 \begin{pmatrix}
  \omega^2 \\ -\omega\bar{I}
 \end{pmatrix} \\
 p_3&: \quad
 \begin{pmatrix}
  \omega \sin(\omega \bar{\eta}) + \bar{I}\cos(\omega \bar{\eta}) & \cos(\omega \bar{\eta}) \\
  \omega \cos(\omega \bar{\eta}) - \bar{I}\sin(\omega \bar{\eta}) & -\sin(\omega \bar{\eta})
 \end{pmatrix}
 \alpha
  =
 \begin{pmatrix}
  \omega^2-\bar{I} \\ -\omega\bar{I}
 \end{pmatrix}
\end{align*}

\begin{figure}
\centering
\includegraphics[width=0.4\columnwidth]{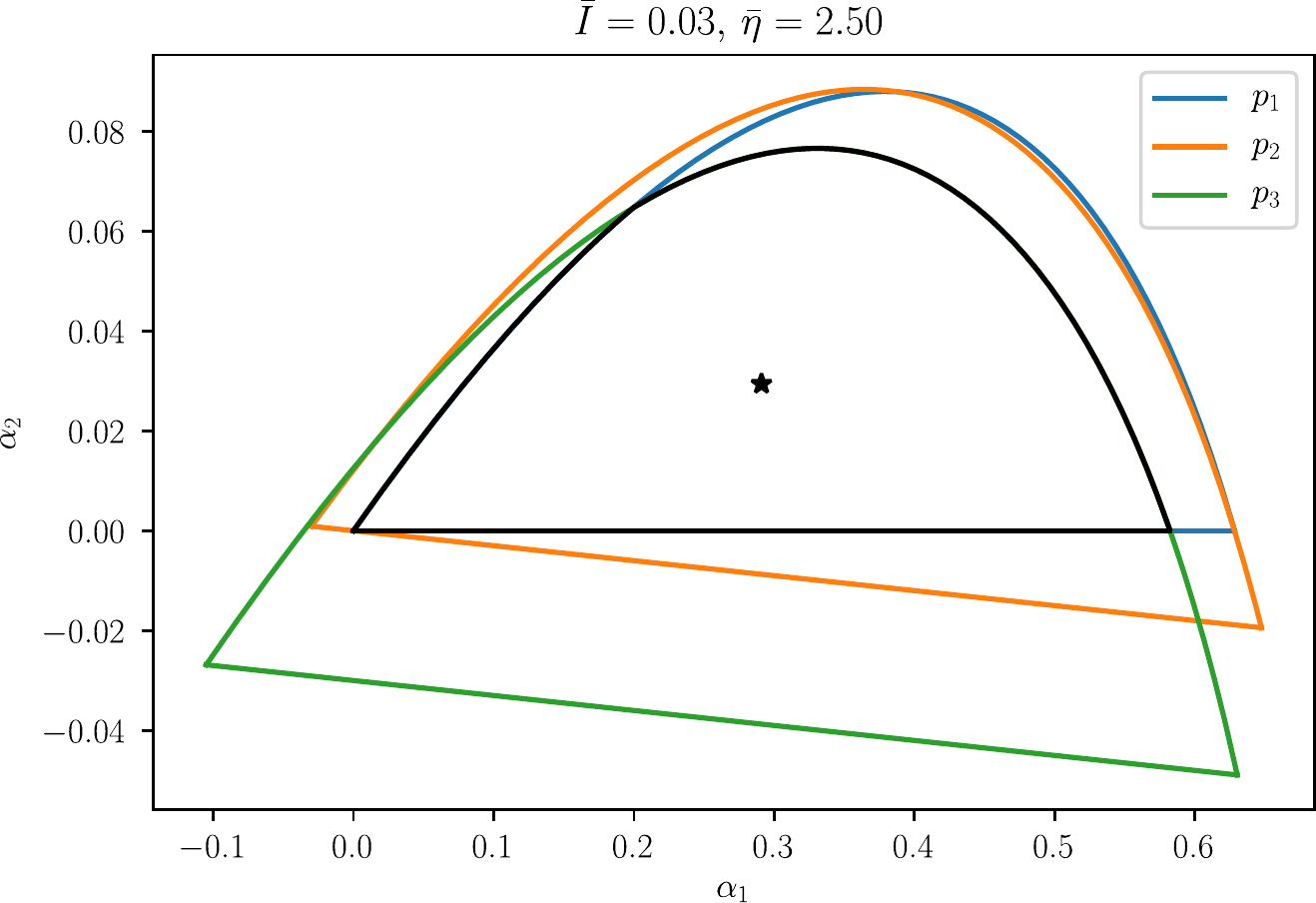}
\caption{Stability/instability boundaries of the quasipolynomials~\eqref{eq:quasipolynomial}
 in the space of parameters $(\alpha_1,\alpha_2)$ for $\bar{I}=0.03$ and $\bar{\eta}=2.5$. The
 boundary of the intersection of the three regions is depicted in black.}
\label{fig:vertices_25}
\end{figure}

The form and disposition of the stability regions depend on $\bar{\eta}$ and $\bar{I}$. 
Figure~\ref{fig:vertices_05} shows their boundaries for a relatively small delay, 
$\bar{\eta} = 0.5$, and the largest incidence, $\bar{I}=1$. Figure~\ref{fig:vertices_25} 
shows the boundaries of the stability regions for a larger delay, $\bar{\eta} = 2.5$,
but a smaller incidence, $\bar{I} = 0.03$. Since stability of the three vertices is a 
necessary condition for the stability of~\eqref{eq:system_tau}-\eqref{eq:polytope}, we require
$\alpha$ to be placed at the intersection of the three regions (boundaries
drawn in black), for example, at an approximate centroid of the intersection (marked with $\star$).

\section{Stability of the predictor} \label{sec:Stab_linear}

Setting $\alpha$ as described in the previous section, i.e., at the
intersection of the stability regions, only ensures that a necessary condition for
stability is satisfied. We will now exploit the polytopic nature of~\eqref{eq:linear_error} 
and the fact that stabilty is ensured by the existence of a Lyapunov-Krasowskii
functional that is common to all the vertices of the polytope.

We will begin by summarizing a general assertion from the book by Fridman~\cite{fridmanE}.

\begin{lemma}
 Consider the candidate Lyapunov-Krasowskii functional
 \begin{equation} \label{eq:V}
  V(\varepsilon_{\tau}) = \varepsilon^{\top}(\tau)P\varepsilon(\tau)+
   \int_{\tau-\bar{\eta}}^{\tau}\varepsilon^{T}(s)S\varepsilon(s)\rd s+\bar{\eta}
   \int_{-\bar{\eta}}^{0}\int_{\tau+\theta}^{\tau}\dot{\varepsilon}^{\top}(s)R\dot{\varepsilon}(s)\rd s \rd\theta
 \end{equation}
 with $P > 0$, $R \ge 0$ and $S \ge 0$. Define
 \begin{displaymath}
  E(\tau) = 
   \begin{pmatrix}
    \varepsilon(\tau) & \dot{\varepsilon}(\tau) & \varepsilon(\tau-\bar{\eta}) & \varepsilon(\tau-\eta(\tau)) 
   \end{pmatrix}^\top \;.
 \end{displaymath}
 The time derivative of $V$ satisfies
 \begin{equation} \label{eq:general_derV}
  \dot{V}(\varepsilon_{\tau}) \le E(\tau)^\top 
   \begin{pmatrix}
    S-R     & P              &  0     & R \\
      \star & \bar{\eta}^2 R &  0     & 0 \\
      \star & \star          & -(S+R) & 0 \\
      \star & \star          &  \star & -2R
   \end{pmatrix}
   E(\tau) \;,
 \end{equation}
 where the terms $\star$ are such that the overall matrix is symmetric.
\end{lemma}

The lemma is easily proved by differentiating $V$ and applying Jensen's lemma~\cite[Ch. 3]{fridmanE}.

\begin{theorem} \label{thm:LMI}
 Let
 \begin{multline} \label{eq:Q}
  Q(C;P,R,S,P_2,P_3) = \\
   \begin{pmatrix}
    C^{T}P_{2}+P_{2}^{T}C+S-R &  P-P_{2}^{T}+C^{T}P_{3}        & 0      & P_{2}^{T}B_{1}+R \\ 
                        \star & -P_{3}-P_{3}^{T}+\bar{\eta}^2R & 0      & P_{3}^{T}B_{1} \\ 
                        \star & \star                          & -(S+R) & R \\ 
                        \star & \star                          & \star  & -2R
   \end{pmatrix} \;,
 \end{multline}
 where $C,P,R,S,P_2$ and $P_3$ are $2\times 2$ matrices. Suppose that there exists
 $P > 0$, $R \ge 0$, $S \ge 0$ and $P_2$, $P_3$ such that 
 \begin{equation} \label{eq:LMI}
  Q(C_i,P,R,S,P_2,P_3) < 0 \;, \quad i=1,2,3 \;.
 \end{equation}
 Then, the trivial solution of the observer error-dynamics~\eqref{eq:system_tau} is asymptotically stable.
\end{theorem}

\begin{proof}
 The proof follows the descriptor approach~\cite{fridmanE}, but we pay special attention
 to the nonlinear terms in~\eqref{eq:system_tau} and incorporate the time-varying nature of
 the system. Consider again~\eqref{eq:V} and note that, for any $P_2, P_3 \in \RE^{2\times 2}$,
 \begin{multline*}
  2\left( \varepsilon(\tau)^\top P_2^\top + \dot{\varepsilon}(\tau)^\top P_3^\top \right) \cdot
   \big( B_0(\tau)\varepsilon(\tau) + B_{1}\varepsilon(\tau-\eta(\tau)) \\
    + H(\tau,\varepsilon(\tau))\varepsilon(\tau) - \dot{\varepsilon}(\tau) \big) = 0
 \end{multline*}
 or, in matrix form,
 \begin{multline} \label{eq:derSys}
    E(\tau)^\top
   \begin{pmatrix}
    B_0(\tau)^\top P_2 + P_2^\top B_0(\tau) & -P_2^\top+B_0(\tau)^\top P_3 & 0     & P_2^\top B_1 \\
    \star                                   & -P_3-P_3^\top                & 0     & P_3^\top B_1 \\
    \star                                   & \star                        & 0     & 0 \\ 
    \star                                   & \star                        & \star & 0 \\
   \end{pmatrix}
    E(\tau) \\ + 
   2\left( \varepsilon(\tau)^\top P_2^\top + \dot{\varepsilon}(\tau)^\top P_3^\top \right) \cdot
    H(\tau,\varepsilon(\tau))\cdot \varepsilon(\tau) = 0 \;.
 \end{multline}
 Observe that
 \begin{displaymath}
  2\left( \varepsilon(\tau)^\top P_2^\top + \dot{\varepsilon}(\tau)^\top P_3^\top \right) \cdot
    H(\tau,\varepsilon(\tau))\cdot \varepsilon(\tau) = \mathcal{O}(\|E(\tau)\|^3) \;,
 \end{displaymath}
 so adding~\eqref{eq:derSys} to~\eqref{eq:general_derV} gives
 \begin{displaymath}
  \dot{V}(\varepsilon_{\tau}) \le -E^\top(\tau) Q(B_0(t);P,R,S,P_2,P_3) E(\tau) + \mathcal{O}(\|E(\tau)\|^3) \;.
 \end{displaymath}
 
 Because of~\eqref{eq:polytope}, there exists continuous real-valued functions $c_i$ such that
 \begin{displaymath}
  B_0(\tau) = \sum_{i=1}^3 c_i(\tau)C_i \;, \quad c_i(\tau) \ge 0 \;, \quad \sum_{i=1}^3 c_i(\tau) \equiv 1 \;.
 \end{displaymath}
 Since $C$ appears affinely in~\eqref{eq:Q}, we have~\cite[Remark~3.6]{fridmanE}
 \begin{displaymath}
  \dot{V}(\varepsilon_{\tau}) \le -\sum_{i=1}^3 c_i(\tau) E^\top(\tau) Q(C_i;P,R,S,P_2,P_3) E(\tau) + 
   \mathcal{O}(\|E(\tau)\|^3) \;.
 \end{displaymath}
 By~\eqref{eq:LMI}, the derivative of $V$ is negative definite in a neighborhood of the origin and
 asymptotic stability follows.
\end{proof}

\begin{figure}
\centering
\includegraphics[width=1.0\columnwidth]{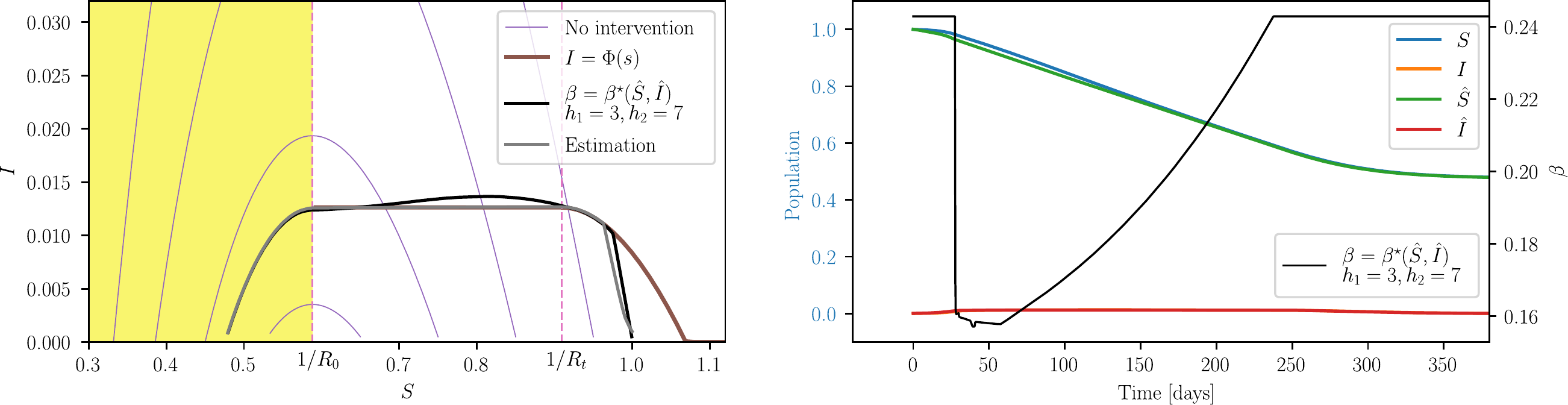}
\caption{True and estimated trajectories under the feedback law $\beta = \beta^\star(\hat{S},\hat{I})$,
 with $h_1 = 3$ and $h_2 = 7$ days.} 
\label{fig:delays_predictor}
\end{figure}

For our example, we take a conservative approach and set
\begin{displaymath}
 \bar{\eta} = 5 > (h_1+h_2)\beta_{\max} = 2.4
\end{displaymath}
and $\bar{I}= 30 \times 10^{-3}$ (recall that $I_{\max} = 12.63 \times 10^{-3}$). 
The gain $\alpha = \begin{pmatrix} 0.115 & 0.005 \end{pmatrix}$ lays within the 
intersection of the stability regions of $p_1$, $p_2$ and $p_3$.

The LMI~\eqref{eq:LMI} was solved using SCS~\cite{scs}. A solution is
\begin{displaymath}
 P = 
  \begin{pmatrix}
   51.1 & -140.6 \\ -140.6 & 979.2
  \end{pmatrix} \;, \quad 
 R = S = 
  \begin{pmatrix}
   16.3 & -0.6 \\ -0.6 & 3.3
  \end{pmatrix}
\end{displaymath}
and
\begin{displaymath}
 P_2 = P_3 = 
  \begin{pmatrix}
     42.3 & -85.3 \\
   -140.5 & 984.4 
  \end{pmatrix} \;,
\end{displaymath}
so the prediction error converges to zero asymptotically. This is illustrated
in Fig.~\ref{fig:delays_predictor}. The hospital capacity is now exceeded by
only 7.8\,\%. 

\section{On the effect of measurement errors} \label{sec:noise}

A frequent situation in epidemics is poor output variable measurement, mainly due 
to underregistration. A sound assumption is that the output is proportional to the measured
variable, $I(t-h_2)$. The proportion may be time-varying but always less than 1. 
It is described as 
\begin{equation*}
 y(t) = I(t-h_2) m(t-h_2)
\end{equation*}
with 
\begin{equation} \label{measurement disturbance}
 m(t) \in [1-\delta,1] \;, \quad  0 \leq \delta < 1 \;.
\end{equation}
The logarithmic term in the prediction error
dynamics~\eqref{observation error general} is now 
\begin{multline*}
 \ln \left( \frac{y(t-h_1)}{\hat{I}(t-h-h_{1})}\right) = 
  \ln \left( \frac{I(t-h)m(t-h)}{\hat{I}(t-h-h_{1})}\right) \\  =
  \ln \left( \frac{I(t-h)}{\hat{I}(t-h-h_{1})}\right) + \ln(m(t-h)) \;.
\end{multline*}
Define $d(t)=\ln (m(t))$ and observe that, by~\eqref{measurement disturbance},
is such that $\left\vert d(t)\right\vert \leq \bar{d}$ with
\begin{displaymath}
 \bar{d} = \ln\left( \frac{1}{1-\delta} \right) \;.
\end{displaymath}
The prediction error now has the dynamics 
\begin{equation} \label{system meas error}
 \dot{\varepsilon}(\tau) = B_0(\tau)\varepsilon(\tau) + B_{1}\varepsilon(\tau-\eta(\tau)) 
 + H(\tau,\varepsilon(\tau))\varepsilon(\tau) 
 - \alpha d(\tau-\eta(\tau)) \;.
\end{equation}

To analyze the effect of the measurement error on the system response, we
compute the time derivative of the functional~\eqref{eq:V}, now along the trajectories of
system (\ref{system meas error}). Following the same steps as in the proof of 
Theorem~\ref{thm:LMI}, we obtain
\begin{multline*}
 \dot{V}(\varepsilon_{\tau}) \le -\sum_{i=1}^3 c_i(\tau) E^\top(\tau) Q(C_i;P,R,S,P_2,P_3) E(\tau) \\
  - 2\left( \varepsilon(\tau)^\top P_2^\top + \dot{\varepsilon}(\tau)^\top P_3^\top \right) 
  \alpha d(\tau-\eta(\tau)) + \mathcal{O}(\|E(\tau)\|^3) \;.
\end{multline*}
From the order relation
\begin{displaymath}
 2\left( \varepsilon(\tau)^\top P_2^\top + \dot{\varepsilon}(\tau)^\top P_3^\top \right)
 \alpha d(\tau-\eta(\tau)) = \mathcal{O}(\bar{d}\cdot\|E(\tau)\|)
\end{displaymath}
we conclude that, for $\bar{d}$ small enough, the solutions are ultimately bounded 
with ultimate bound proportional to $\bar{d}$.

We now simulate the effects of the measurement noise described above.
We take $\delta = 0.5$ and define $m(t)$ a random variable with a beta
distribution. We perform simulations for the predictor and the observer 
(see Fig.~\ref{fig:delays_noise}). 
The measurement errors result in a higher number of infected people. 
The poor performance of the observer had already been established. 
Noise simply makes it worse. In the case of the predictor, the increment 
is relatively low if we take into account the high amplitude of the error.

\begin{figure}
\centering
\includegraphics[width=1.0\columnwidth]{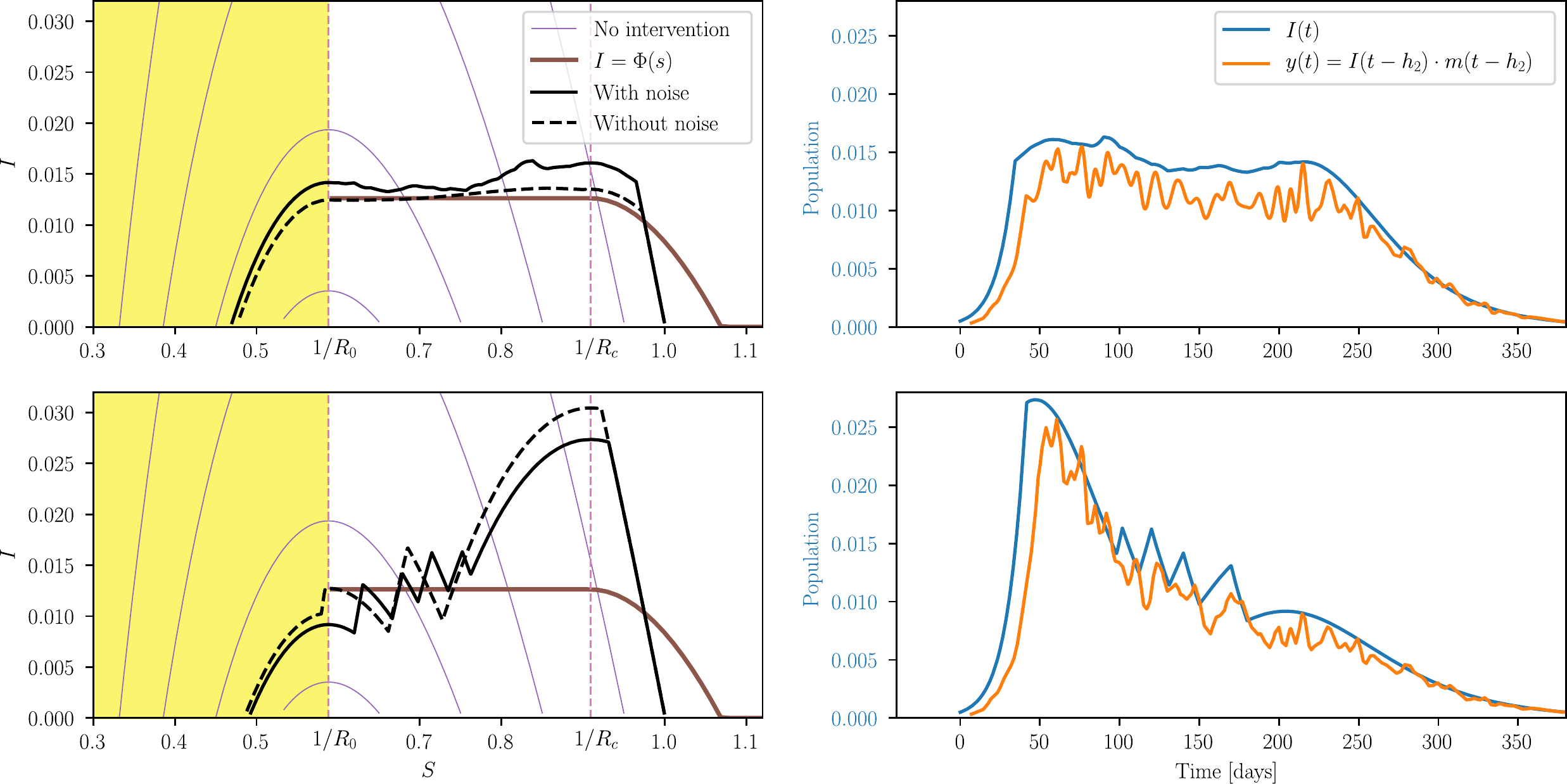}
\caption{Closed-loop trajectories using the predictor~\eqref{eq:predictor} (top)
 and the observer~\eqref{eq:observer} (bottom). Simulations were performed with and without measurement errors.}
\label{fig:delays_noise}
\end{figure}

\section{Conclusions}

We have presented a predictor for systems with large input delays. The predictor 
is designed using the observer-predictor methodology introduced in the past years.
In contrast with existing proposals, we present simple tuning criteria for ensuring
the simultaneous stability of the polytopic model's vertices. The stability of the
complete polytopic model is then established formally with the help of a
Lyapunov-Krasovskii functional. The functional can also be used to assess the sensitivity
of the predictor with respect to error measurements. It is worthy of mentioning that
it can also be used to find estimates of the domain of attraction, robustness bounds
for the delay, and parameter uncertainties, among other problems of interest.

Our current research concerns the extension of the present results, obtained for
two-dimensional systems, to $n$-dimensional ones, and to apply them to the study of
more elaborated epidemic models involving states such as vaccinated individuals and
asymptomatic ones.

\bibliographystyle{plain}
\bibliography{references}

\end{document}